\documentclass[11pt]{article}

\usepackage[margin=0.9in]{geometry}

\usepackage{graphicx}

\usepackage{amssymb}
\usepackage{amsmath}
\usepackage{amsthm}
\usepackage{mathrsfs}

\usepackage{xspace}

\allowdisplaybreaks

\theoremstyle{definition}
\newtheorem{thm}{Theorem}[section]

\newtheorem{proposition}[thm]{Proposition}

\newtheorem{remark}[thm]{Remark}

\numberwithin{thm}{section}

\newcommand{\mR}{\mathbb{R}}
\newcommand{\mP}{\mathcal{P}}

\newcommand{\mG}{\mathcal{G}}
\newcommand{\mN}{\mathcal{N}}

\newcommand{\mAD}{\text{AD}\xspace}
\newcommand{\mSH}{\text{Sh}\xspace}

\newcommand{\mCE}{\textbf{CE}\xspace}
\newcommand{\mCS}{\textbf{CS}\xspace}
\newcommand{\mADD}{\textbf{ADD}\xspace}
\newcommand{\mZOT}{\textbf{ZOT}\xspace}
\newcommand{\mFOT}{\textbf{FOT}\xspace}
\newcommand{\mGN}{\textbf{GN}\xspace}
\newcommand{\mSP}{\textbf{SP}\xspace}
\newcommand{\mOU}{\textbf{OU}\xspace}

\newcommand{\mZG}{\textbf{ZG}\xspace}
\newcommand{\mOE}{\textbf{OE}\xspace}
\newcommand{\mZI}{\textbf{ZI}\xspace}
\newcommand{\mEI}{\textbf{EI}\xspace}

\begin{document}
\title{How outside options are incorporated into payoff distributions}
\author{Takaaki Abe\thanks{Department of Economic Engineering, School of Economics, Kyushu University, 744, Motooka, Nishi-ku, Fukuoka, 819-0395, Japan. Email: takaakiabe@econ.kyushu-u.ac.jp
}}
\date{}
\maketitle

\begin{abstract}
This paper examines how outside options are incorporated into payoff distributions in games with coalition structures. We introduce and characterize the $\alpha$-value, which ``fully'' incorporates outside options, and provide a new characterization of the Aumann-Dr\`{e}ze value as an allocation rule that does not incorporate outside options. We show that the $\chi$-value (Casajus, 2009) is a component-wise convex combination of these two values and thus incorporates outside options in a discounted form.
\end{abstract}

\noindent Keywords: allocation rule; coalition structure; null player; outside option\\ 
\noindent JEL Classification: C71

\section{Introduction}\label{SEC_INTRO}

Casajus (2009) considered the following problem, commonly known as a \textit{gloves game}. Suppose that there are two left-glove holders, $L_1$ and $L_2$, and four right-glove holders, $R_1, ..., R_4$. They are organized into the following four components: 
\[
\{\{L_1,R_1\},\{L_2,R_2\},\{R_3\},\{R_4\}\}.
\]
A matching pair generates a worth of one, whereas a singleton generates no worth. How should the two players in a matching pair divide the worth of one? Perhaps the best-known allocation rule for this problem is the \textit{Aumann-Dr\`{e}ze value} (Aumann and Dr\`{e}ze, 1974), hereafter the \mAD-value. The \mAD-value divides the worth of one equally between $L_1$ and $R_1$. This is because, if the other four players are disregarded, $L_1$ and $R_1$ are symmetric within their two-player component. When all six players are taken into account, however, they are no longer symmetric. Indeed, left gloves are scarcer than right gloves. Player $L_1$ can therefore argue to her partner $R_1$ that she has alternative options outside their component, namely, forming a matching pair with either $R_3$ or $R_4$. The theory of allocation rules for games with coalition structures captures such possibilities through the concept of \textit{outside options}.

Within the framework of games with coalition structures, Wiese (2007) introduced an allocation rule that takes outside options into account. In the gloves game described above, the Wiese value assigns $43/60$ $(=0.7166\ldots)$ to player $L_1$ and $17/60$ $(=0.2833\ldots)$ to player $R_1$. The fact that $L_1$ receives more than $R_1$ reflects that the Wiese value incorporates the outside options available to $L_1$. While acknowledging that the Wiese value incorporates outside options, Casajus (2009) argued that its axiomatic characterization relies on an unappealing axiom. He introduced an alternative axiom, called \textit{Splitting}, and proposed a new allocation rule, the \textit{$\chi$-value}, together with its axiomatic characterization. The $\chi$-value has several desirable features. It takes outside options into account, admits an intuitive axiomatic characterization, and guarantees the existence of stable coalition structures.

Since its introduction by Casajus (2009), the $\chi$-value has become a prominent allocation rule for incorporating outside options and has been studied from various perspectives.
Tuti\'{c} et al. (2011) experimentally compared the $\chi$-value with other allocation rules. Hiller (2011) generalized the Shapley-value term in the formula of the $\chi$-value, while Hiller (2015) applied the $\chi$-value to several special classes of games. Casajus and Tuti\'{c} (2013) introduced asymmetric weights into the allocation rule from a Nash bargaining perspective. Abe (2021) studied the relationship between the Splitting axiom and the class of simple coalition formation games. Yu et al. (2023) analyzed a new allocation rule that incorporates the Owen value (Owen, 1977) into the functional form of the $\chi$-value. Casajus and La Mura (2024) extended the Splitting axiom by introducing the Relative Splitting property. Abe (2024) used the Splitting axiom to characterize allocation rules based on equal surplus sharing.

The starting point of our research is a simple question: when the $\chi$-value is said to incorporate outside options, how, and to what extent, does the value incorporate them into payoff distributions? The existing literature establishes that outside options matter under the allocation rule, but neither its formula nor its axiomatic characterization reveals how outside option claims enter a player's payoff or how strongly they are reflected in the resulting allocation. We answer this question by decomposing a player's payoff into two distinct parts. The first is a baseline allocation determined by the player's within-component contributions. The second is an adjustment based on the outside option claims of the members of the component. This decomposition makes it possible to study each player's outside option claim independently of their within-component productivity.

A key step in our analysis is to extend the \textit{null player out} property (Derks and Haller, 1999). Because the departure of a null player from the game leaves the within-component contributions of the remaining players unchanged, any resulting payoff change reveals how the allocation rule incorporates outside option claims. Building on this observation, we introduce and characterize a new CS-value that ``fully'' incorporates outside option claims, while the \mAD-value serves as the opposite benchmark that does not incorporate them. We show that the $\chi$-value can be expressed as a convex combination of these two extreme benchmarks, showing that the $\chi$-value incorporates outside option claims in a discounted form.

The remainder of this paper is organized as follows. Section \ref{SEC_PREL} introduces the definitions and the basic framework. Section \ref{SEC_OOCLAIM} examines and justifies the standard outside option claims implicitly contained in the formula of the $\chi$-value. Section \ref{SEC_CHAR} provides an axiomatic analysis of how allocation rules incorporate outside option claims. Section \ref{SEC_CONC} summarizes the main findings and presents a table of the characterization results. The proofs are presented in Section \ref{SEC_PROOFS}. The independence of the axioms is provided in Section \ref{SEC_INDEP}.

\section{Preliminaries}\label{SEC_PREL}
\subsection{Games with coalition structures}
Let $\mathcal{U}$ be a countably infinite set of players.
A \textit{coalition} is a finite subset of $\mathcal{U}$.
Let $\mN$ be the set of all finite nonempty subsets of $\mathcal{U}$.
For every $N\in \mN$, let $\mG(N)=\{ v \mid v:2^N \rightarrow \mR, v(\emptyset)=0\}$ be the set of all characteristic functions on $N$, and let $\Pi(N)$ be the set of all partitions of $N$, also called \textit{coalition structures}.
Each element of a coalition structure $\mP \in \Pi(N)$ is called a \textit{component}.
Let $\Gamma=\{ (N,v,\mP) \mid N\in \mN, v\in \mG(N), \mP\in \Pi(N) \}$ denote the set of all \textit{games with coalition structures} (\textit{CS-games}).
A \textit{CS-value} is a function $f$ that maps each $(N,v,\mP)\in \Gamma$ to a payoff vector $f(N,v,\mP)\in \mR^N$.
For every finite set $X$, let $|X|$ denote the cardinality of $X$.

Let $N\in \mN$ and $\mP\in \Pi(N)$.
For every $i\in N$, let $\mP(i)$ denote the component of $\mP$ that contains $i$.
For every $T\subseteq N$, let $\mP(T)=\{C\in \mP \mid C\cap T \neq \emptyset \}$ denote the set of components of $\mP$ that intersect $T$.

For every $N\in \mN$ with $|N|\geq 2$, $v\in \mG(N)$, $\mP\in \Pi(N)$, and $k\in N$, define $v^{-k} \in \mG(N\setminus \{k\})$ as follows: $v^{-k}(S)=v(S)$ for every $S\subseteq N\setminus \{k\}$.
Moreover, $\mP^{-k} \in \Pi(N\setminus \{k\})$ is defined by $\mP^{-k}=\{C\setminus \{k\} \mid C\in \mP,\ C\setminus \{k\} \neq \emptyset\}$.
Similarly, for every $(N,v,\mP)\in\Gamma$ and $K\subsetneq N$, define $v^{-K}(S)=v(S)$ for every $S\subseteq N\setminus K$, and $\mP^{-K}=\{C\setminus K \mid C\in \mP,\ C\setminus K \neq \emptyset\}$.
For notational simplicity, for every $(N,v,\mP)\in\Gamma$ and nonempty $K\subseteq N$, define $v|_{K}\in \mG(K)$ as follows: $v|_{K}(S)=v(S)$ for every $S\subseteq K$.
Similarly, let $\mP|_{K}= \{C\cap K \mid C\in \mP,\ C\cap K\neq \emptyset\} \in \Pi(K)$.

Let $N\in \mN$ and $v\in \mG(N)$.
A player $i\in N$ is a \textit{null player} in $v$ if $v(S\cup \{i\})-v(S)=0$ for every $S\subseteq N\setminus \{i\}$.
Players $i,j\in N$ are \textit{symmetric} in $v$, written $i \overset{v}{\sim} j$, if $v(S\cup \{i\})=v(S\cup \{j\})$ for every $S\subseteq N\setminus \{i,j\}$.
For every nonempty $T\subseteq N$, the \textit{$T$-unanimity game} $u_T\in\mG(N)$ is defined as follows: for every $S\subseteq N$,
\[
u_T(S)=\begin{cases}
1 & \text{if } T\subseteq S,\\
0 & \text{otherwise}.
\end{cases}
\]
For every $v\in\mG(N)$ and nonempty $T\subseteq N$, let $\lambda^v_T=\sum_{R\subseteq T}(-1)^{|T|-|R|}v(R)$,
where $\lambda^v_T$ is called the \textit{Harsanyi dividend} associated with coalition $T$. 
Every characteristic function $v\in\mG(N)$ is uniquely expressed as $v=\sum_{\emptyset \neq T\subseteq N} \lambda^v_T u_T$.

The \textit{Shapley value} is defined as follows: for every $N\in \mN$, $v\in \mG(N)$, and $i\in N$,
\[
\mSH_i(N,v)=\sum_{S\subseteq N\setminus\{i\}}
\frac{|S|!(|N|-|S|-1)!}{|N|!} \left( v(S\cup\{i\})-v(S) \right).
\]
The \textit{Aumann-Dr\`{e}ze value} is defined as follows: for every $(N,v,\mP)\in\Gamma$ and $i\in N$,
\[
\mAD_i(N,v,\mP)=\mSH_i(\mP(i),v|_{\mP(i)}).
\]
Casajus (2009) introduced the \textit{$\chi$-value} as follows: for every $(N,v,\mP)\in\Gamma$ and $i\in N$,
\[
\chi_i(N,v,\mP)
=\mSH_i(N,v)+\frac{1}{|\mP(i)|}(v(\mP(i))-\sum_{j\in \mP(i)}\mSH_j(N,v)).
\]

\subsection{Outside option claims in the $\chi$-value}
The $\chi$-value can be rewritten in the following equivalent form:
\[
\chi_i(N,v,\mP)
=\mAD_i(N,v,\mP)+\frac{1}{|\mP(i)|}\left( |\mP(i)|\cdot \omega^*_i(N,v,\mP)-\sum_{j\in \mP(i)}\omega^*_j(N,v,\mP) \right),
\]
where
\[
\omega^*_i(N,v,\mP)=\mSH_i(N,v)-\mAD_i(N,v,\mP).
\]
This formula consists of two parts: (1) a baseline allocation based on within-component contributions and (2) an adjustment based on outside option claims.
The first part is $\mAD_i(N,v,\mP)$, which aggregates player $i$'s contributions within $\mP(i)$ and does not account for interactions between player $i$ and players outside $\mP(i)$.

To interpret the second part, consider first the expression $|\mP(i)|\cdot \omega^*_i(N,v,\mP)$. 
We call $\omega^*_i(N,v,\mP)$ the \textit{standard outside option claim} of player $i$.
The expression above describes a system of ``transfers'' among the members of $\mP(i)$ based on their standard outside option claims.
The term $|\mP(i)|\cdot \omega^*_i(N,v,\mP)$ indicates that player $i$ receives $\omega^*_i$ from each member of $\mP(i)$, whereas $\sum_{j\in \mP(i)}\omega^*_j(N,v,\mP)$ indicates that player $i$ pays $\omega^*_j$ to each member $j$ of $\mP(i)$. 
This is equivalent to $\sum_{j\in \mP(i)}(\omega^*_i(N,v,\mP)-\omega^*_j(N,v,\mP))$.
Thus, the expression represents player $i$'s net transfer based on the outside option claims of the component members.
The factor $1/|\mP(i)|$ in the formula for the $\chi$-value discounts this net transfer by the size of the component. Therefore, the second part of the formula can be interpreted as an adjustment based on outside option transfers, discounted by the number of members of $\mP(i)$.
Given this interpretation, we pose the following questions:
\begin{itemize}
\item Is $\omega^*$ embedded in the $\chi$-value an appropriate measure of outside option claims?
\item Are there alternatives to the $\chi$-value for aggregating outside option claims?
\end{itemize}
We address the first question in Section \ref{SEC_OOCLAIM} and the second in Section \ref{SEC_CHAR}.

\section{Measuring outside option claims}\label{SEC_OOCLAIM}

\subsection{Standard outside option claim rule}
An \textit{outside option claim rule} (\textit{OO-rule}) is a function $\omega$ that maps each $(N,v,\mP)\in\Gamma$ to a vector $\omega(N,v,\mP)\in\mR^N$.
We call $\omega_i(N,v,\mP)$ the outside option claim of player $i$ in the CS game $(N,v,\mP)$.
As mentioned in the preceding section, the standard OO-rule $\omega^*$ is given as follows: for every $(N,v,\mP)\in\Gamma$ and $i\in N$,
\[
\omega_i^*(N,v,\mP)=\mSH_i(N,v)-\mAD_i(N,v,\mP),
\]
where $\mSH_i(N,v)$ evaluates player $i$'s contributions to all coalitions without imposing the restrictions by the coalition structure, while $\mAD_i(N,v,\mP)$ aggregates player $i$'s contributions within the game restricted to $\mP(i)$ and therefore disregards all interactions between player $i$ and players outside $\mP(i)$.
Therefore, this difference can be interpreted as the value of player $i$'s inter-component contributions involving players outside $\mP(i)$.

The standard OO-rule has the following equivalent representation: for every $(N,v,\mP)\in\Gamma$ and $i\in N$,
\[
\omega^*_i(N,v,\mP)=\sum_{\substack{T\subseteq N \\ i\in T,\ |\mP(T)|\geq 2}}\frac{\lambda^v_T}{|T|}.
\]
The Harsanyi dividend $\lambda^v_T$ represents the net surplus, or synergy, generated by $T$ that cannot be attributed to any proper subcoalition of $T$.
The condition $|\mP(T)|\geq 2$ means that $T$ intersects at least two components of $\mP$. 
Since $i\in T$ implies that $T$ intersects $\mP(i)$, any coalition $T$ satisfying $|\mP(T)|\geq 2$ and $i\in T$ contains at least one player outside $\mP(i)$.
Therefore, $\omega_i^*(N,v,\mP)$ aggregates player $i$'s shares of the net surpluses generated by all inter-component coalitions $T$ containing $i$.

The standard OO-rule also admits an alternative interpretation in terms of a game that separates inter-component interactions from within-component interactions. 
Let $(N,v,\mP)\in \Gamma$. For every nonempty $C\subseteq N$, define $v_{[C]}\in \mG(N)$ as follows: for every $S\subseteq N$, $v_{[C]}(S):=v(S\cap C)$.
For every $C\in\mP$, we call $(N, v_{[C]})$ the \textit{$C$-inside game} associated with $(N,v,\mP)$. The $C$-inside game embeds the game restricted to $C$ (i.e., $v|_C$) in the player set $N$.
We then define 
\[
v^\mP_{\text{out}}=v- \sum_{C\in \mP}v_{[C]}
\]
and call $(N, v^\mP_{\text{out}})$ the \textit{outside game} induced by $(N,v,\mP)$.
By subtracting the inside games associated with all components of $\mP$ from the original game $v$, the outside game captures only the inter-component interactions.
Then, the standard OO-rule is expressed as 
\[
\omega^*(N,v,\mP)=\mSH(N,v^\mP_{\text{out}}).
\]
This representation shows that $\omega^*$ distributes the total worth $v^\mP_{\text{out}}(N)$ generated by all inter-component interactions among the players.

\subsection{Characterization of the standard OO-rule}
In addition to the functional representations presented above, we now examine the standard OO-rule from an axiomatic perspective.
Let $\omega$ be an OO-rule.
For every $N\in \mN$, let $\textbf{0}\in \mG(N)$ denote the game defined by $\textbf{0}(S)=0$ for every $S\subseteq N$. We call $\textbf{0}$ the \textit{zero game}.
The following property states that if no coalition generates any worth, then every player's outside option claim is zero.

\begin{itemize}
\item[] \textbf{Zero-Game (\mZG).} For every $N\in\mN$, $\mP\in\Pi(N)$, and $i\in N$, $\omega_i(N,\textbf{0},\mP)=0$.
\end{itemize}

The next property requires that the total value attributable to inter-component interactions be fully accounted for by the outside option claims of all players. The quantity $v(N)-\sum_{C\in\mP}v(C)\ (=v^\mP_{\mathrm{out}}(N))$ represents the total value generated by inter-component interactions under $\mP$.
\begin{itemize}
\item[] \textbf{Outside Option Efficiency (\mOE).} 
For every $(N,v,\mP)\in\Gamma$,
$\sum_{j\in N}\omega_j(N,v,\mP)=v(N)-\sum_{C\in\mP}v(C)$.
\end{itemize}

The next two properties describe how outside option claims respond to changes in the net surplus attributable to a particular coalition. They are adapted, respectively, from Coalitional Strategic Equivalence and Fair Ranking, both introduced by Chun (1989, 1991).

Adding $\lambda u_T$ to $v$ changes the net surplus associated with coalition $T$ by $\lambda$, without changing the net surplus associated with any other coalition.
Chun's Coalitional Strategic Equivalence states that such a change specific to $T$ is confined to the members of $T$ and does not affect the payoffs of players in $N\setminus T$.
Applying this principle to outside option claims, Zero Increment requires that a change in the net surplus associated with $T$ does not affect the outside option claim of any player in $N\setminus T$.

\begin{itemize}
\item[] \textbf{Zero Increment (\mZI).} 
For every $(N,v,\mP)\in\Gamma$, nonempty $T\subseteq N$, $\lambda\in\mR$, and $i\in N\setminus T$,
$\omega_i(N,v+\lambda u_T,\mP)=\omega_i(N,v,\mP)$.
\end{itemize}

Equal Increment adapts the fairness principle underlying Chun's Fair Ranking to outside option claims. The property states that every member of $T$ experiences the same change in their outside option claim when the net surplus associated with $T$ changes.

\begin{itemize}
\item[] \textbf{Equal Increment (\mEI).} 
For every $(N,v,\mP)\in\Gamma$, nonempty $T\subseteq N$, $\lambda\in\mR$, and $i,j\in T$,
$\omega_i(N,v+\lambda u_T,\mP)-\omega_i(N,v,\mP)=\omega_j(N,v+\lambda u_T,\mP)-\omega_j(N,v,\mP)$.
\end{itemize}

These properties yield the following characterization of the standard OO-rule.
\begin{proposition}\label{PROP_OMEGA}
An OO-rule $\omega$ satisfies \mZG, \mOE, \mZI, and \mEI if and only if $\omega=\omega^*$.
\end{proposition}

\begin{remark}
One might seek an alternative characterization of the standard OO-rule by combining \mOE, Additivity, Null Player, and Symmetry.\footnote{For OO-rules, these properties are defined as follows.
Additivity: For every $N\in \mN$, $\mP\in\Pi(N)$, and $v,w\in\mG(N)$, $\omega(N,v+w,\mP) = \omega(N,v,\mP)+\omega(N,w,\mP)$.
Null Player: For every $(N,v,\mP)\in\Gamma$ and $i\in N$, if $i$ is a null player in $v$, then $\omega_i(N,v,\mP)=0$.
Symmetry: For every $(N,v,\mP)\in\Gamma$ and $i,j\in N$, if $i \overset{v}{\sim} j$, then $\omega_i(N,v,\mP)=\omega_j(N,v,\mP)$.
}
This approach, however, does not work because $\omega^*$ incorporates the \mAD-value, which violates Symmetry.

One may instead weaken Symmetry to Coalitional Symmetry, which requires symmetry only for symmetric players belonging to the same component (see Section \ref{SEC_CHAR} for the formal definition).
The standard OO-rule $\omega^*$ satisfies this weaker property. However, Coalitional Symmetry is satisfied by a variety of OO-rules.
For example, consider the following OO-rule $\widehat{\omega}$:
\[
\widehat{\omega}_i(N,v,\mP)
=\sum_{\substack{T\subseteq N \\ i\in T,\ |\mP(T)|\geq 2}}\frac{\lambda_T^v}{|\mP(T)|\,|T\cap\mP(i)|}.
\]
For each inter-component coalition $T$, this rule first divides the Harsanyi dividend equally among the components in $\mP(T)$ and then divides each component's share equally among the members of $T$ belonging to that component. The rule $\widehat{\omega}$ and various alternatives satisfy \mOE, Additivity, Null Player, and Coalitional Symmetry.
In this sense, Symmetry is too strong to be satisfied by the standard OO-rule, whereas Coalitional Symmetry is too weak to characterize it.
\end{remark}

\section{CS-values and discounted outside option claims}\label{SEC_CHAR}

\subsection{Axioms for CS-values}
In this section, we analyze how CS-values aggregate outside option claims. In particular, we introduce a CS-value that ``fully'' incorporates the standard outside option claims of all players and show that the $\chi$-value incorporates these claims in a discounted form. The $\chi$-value satisfies the following three basic properties.

\begin{itemize}
\item[] \textbf{Component Efficiency (\mCE).} For every $(N,v,\mP)\in \Gamma$ and $C\in \mP$, $\sum_{j\in C} f_j(N,v,\mP)=v(C)$.

\item[] \textbf{Component Symmetry (\mCS).} For every $(N,v,\mP)\in \Gamma$ and $i,j\in N$, if $\mP(i)=\mP(j)$ and $i \overset{v}{\sim} j$, then $f_i(N,v,\mP)=f_j(N,v,\mP)$.

\item[] \textbf{Additivity (\mADD).} For every $N\in \mN$, $\mP\in\Pi(N)$, and $v,w\in\mG(N)$, $f(N,v+w,\mP)=f(N,v,\mP)+f(N,w,\mP)$.
\end{itemize}

These properties do not specify how a CS-value takes outside option claims into account. \mCE states that the total payoff assigned to the members of a component is equal to the worth of the component. \mCS requires equal treatment of symmetric players who belong to the same component. This requirement is also known as Restricted Equal Treatment. \mADD specifies how a CS-value responds to the addition of games.
In addition to these properties, the $\chi$-value satisfies the following weakened version of the null player property, which requires a null player to receive zero when the coalition structure is the grand coalition.
\begin{itemize}
\item[] \textbf{Grand Coalition Null Player (\mGN).} For every $N\in \mN$, $v\in\mG(N)$, and $i\in N$, if $i$ is a null player in $v$, then $f_i(N,v,\{N\})=0$.
\end{itemize}

The final property used to characterize the $\chi$-value is \textit{Splitting}. This property concerns changes in payoffs resulting from a refinement of the coalition structure.
Let $\mP,\mP' \in \Pi(N)$. We say that $\mP$ is \textit{finer than} $\mP'$, written $\mP \sqsubseteq \mP'$, if for every $i\in N$, $\mP(i)\subseteq \mP'(i)$.
\begin{itemize}
\item[] \textbf{Splitting (\mSP).} For every $N\in \mN$, $v\in\mG(N)$, $\mP,\mP' \in\Pi(N)$, and $i,j\in N$, if $\mP\sqsubseteq \mP'$ and $\mP(i)=\mP(j)$, then $f_i(N,v,\mP)-f_i(N,v,\mP')=f_j(N,v,\mP)-f_j(N,v,\mP')$.
\end{itemize}
Splitting requires that, when a coalition structure is refined, the payoffs of any two players who remain together in the same component change by the same amount.
Casajus (2009) showed that the $\chi$-value is the unique CS-value satisfying \mCE, \mCS, \mADD, \mGN, and \mSP.
Casajus (2009) introduced Splitting as an intuitively appealing alternative to Wiese's (2007) Outside Options for Unanimity Games property, defined as follows.
\begin{itemize}
\item[] \textbf{Outside Options for Unanimity Games (\mOU).}
For every $N\in\mN$, $\mP\in\Pi(N)$, $C\in\mP$, and nonempty $T\subseteq N$,
\[
\sum_{j\in C\setminus T}f_j(N,u_T,\mP) =
\begin{cases}
0 & \text{if }|\mP(T)|=1, \\
-\frac{|C\cap T|}{|T|} \cdot \frac{|C\setminus T|}{|C\cup T|} & \text{if }|\mP(T)|>1.
\end{cases}
\]
\end{itemize}
Wiese (2007) characterized the \textit{Wiese-value} (the $W$-value) by \mCE, \mCS, \mADD, and \mOU.\footnote{Wiese (2007) originally imposed Linearity. As Casajus (2009) observed, however, \mADD is sufficient.}
The $W$-value is defined as follows: for every $(N,v,\mP)\in\Gamma$ and $i\in N$,
\[
W_i(N,v,\mP)=\frac{1}{|\Sigma(N)|}\sum_{\sigma\in\Sigma(N)}
\begin{cases}
v(\mP(i)) -\sum_{j\in\mP(i)\setminus\{i\}} [v(P_j^\sigma \cup\{j\})-v(P_j^\sigma)] &\text{if }\sigma\in\Sigma_i(N,\mP),\\
v(P_i^\sigma \cup\{i\})-v(P_i^\sigma) &\text{if }\sigma\notin\Sigma_i(N,\mP),
\end{cases}
\]
where $\Sigma(N)$ represents the set of orderings of $N$, $P_i^\sigma$ denotes the set of predecessors of player $i$ in ordering $\sigma\in\Sigma(N)$, and $\Sigma_i(N,\mP) =\{ \sigma\in\Sigma(N) \mid \sigma(j)<\sigma(i) \text{ for every } j\in\mP(i)\setminus\{i\} \}$ is the set of orderings in which player $i$ is the last member of $\mP(i)$.

\subsection{Separating within-component contributions and outside option claims}
The $\chi$-value was introduced as a CS-value that incorporates outside options. Our question is how, and to what extent, it incorporates them into payoff distributions.
The axioms used to characterize the $\chi$-value do not directly specify how each player's outside option claim is incorporated into the player's payoff.
In particular, \mSP describes how a refinement of the coalition structure affects a player's payoff. The resulting payoff change reflects both a change in the player's within-component contributions and a change in the player's outside option claim.
Thus, \mSP captures these two effects jointly rather than separating them.

To separate the effect of outside option claims from the effect of within-component contributions, we focus on the departure of a null player from the game. This approach adapts the \textit{Null Player Out} property of Derks and Haller (1999) to the analysis of  outside option claims.\footnote{Related axiomatic analyses involving the removal of players in games \textit{without} coalition structures include Kamijo and Kongo (2010) and Kongo (2024).}
If a null player in a component leaves the game, neither the productive contributions of the remaining players nor the worth of the component changes. However, the remaining players have one fewer component member against whom they can assert their outside option claims.
Therefore, if a CS-value takes those claims into account, the remaining players' payoffs may change. Since their within-component contributions remain unchanged, the resulting payoff changes provide a means of identifying how the CS-value responds to their outside option claims.

To formalize this observation, let $k\in N$ be a null player in $v$, and let $C=\{k,i_1,...,i_m\}\in\mP$ with $m\geq 1$.
As suggested by the transfer interpretation of the $\chi$-value in Section \ref{SEC_PREL}, suppose that each member of $C$ asserts an outside option claim against each of the other members of $C$.
Consider player $i_1$. Since $k$ is a null player, removing $k$ from the game changes neither $i_1$'s productive contributions within $C$ nor the worth of the component. The relevant change is that $i_1$ can no longer assert her outside option claim against $k$.
This motivates using the payoff difference $f_{i_1}(N,v,\mP)-f_{i_1}(N\setminus \{k\},v^{-k},\mP^{-k})$ to measure the extent to which the CS-value $f$ incorporates player $i_1$'s outside option claim in her payoff.

The following two axioms impose contrasting requirements on how a CS-value incorporates outside option claims.
The first axiom states that a CS-value does not incorporate outside option claims.
\begin{itemize}
\item[] \textbf{Zero Outside Option Transfer (\mZOT).} Let $(N,v,\mP)\in \Gamma$ with $|N|\geq 2$. If $k\in N$ is a null player in $v$ and $|\mP(k)|\geq 2$, then for every $i\in \mP(k)\setminus \{k\}$,
\[
f_i(N,v,\mP)-f_i(N\setminus \{k\},v^{-k},\mP^{-k})=0.
\]
\end{itemize}

The second axiom states that a CS-value $f$ ``fully'' incorporates outside option claims. To explain this requirement, consider the two terms on the right-hand side of the equality below. The first term, $f_i(N,v,\{N\})$, incorporates all transactions involving player $i$, including those with players outside $\mP(i)$.
In contrast, the second term, $f_i(\mP(i),v|_{\mP(i)},\{\mP(i)\})$, incorporates only transactions between player $i$ and the other members of $\mP(i)$.
Hence, the difference, $f_i(N,v,\{N\})-f_i(\mP(i),v|_{\mP(i)},\{\mP(i)\})$, represents the aggregate value of transactions involving player $i$ and at least one player outside $\mP(i)$.

\begin{itemize}
\item[] \textbf{Full Outside Option Transfer (\mFOT).} Let $(N,v,\mP)\in \Gamma$ with $|N|\geq 2$. If $k\in N$ is a null player in $v$ and $|\mP(k)|\geq 2$, then for every $i\in \mP(k)\setminus \{k\}$,
\[
f_i(N,v,\mP)-f_i(N\setminus \{k\},v^{-k},\mP^{-k})=f_i(N,v,\{N\})-f_i(\mP(i),v|_{\mP(i)},\{\mP(i)\}).
\]
\end{itemize}

The following proposition clarifies the relationship between \mGN and these requirements.
\begin{proposition}\label{PROP_CE_FOT_GN}
If a CS-value $f$ satisfies \mCE and either \mFOT or \mZOT, then $f$ satisfies \mGN.
\end{proposition}

In the characterization of the $\chi$-value by Casajus (2009), \mGN plays only a limited role: together with \mCE, \mCS, and \mADD, it serves to establish that the CS-value coincides with the Shapley value when $\mP=\{N\}$.
Although Wiese (2007) does not directly use \mGN to characterize the $W$-value, the $W$-value also satisfies \mGN, as it coincides with the Shapley value when $\mP=\{N\}$. The \mAD-value also satisfies \mCE, \mCS, \mADD, and \mGN.

These observations indicate that \mCE, \mCS, \mADD, and \mGN determine a CS-value when the coalition structure is the grand coalition but do not determine how the CS-value incorporates outside option claims under a general coalition structure.
In the next subsection, we introduce a new CS-value and show that \mFOT and \mZOT distinguish between contrasting ways of incorporating outside option claims while holding within-component contributions fixed.

\subsection{The $\alpha$-value}
We define the \textit{$\alpha$-value} as follows: for every $(N,v,\mP)\in \Gamma$ and $i\in N$,
\[
\alpha_i(N,v,\mP)= \mAD_i(N,v,\mP) + |\mP(i)| \cdot \omega^*_i(N,v,\mP) - \sum_{j\in \mP(i)}\omega^*_j(N,v,\mP),
\]
where $\omega^*_j(N,v,\mP)=\mSH_j(N,v)-\mAD_j(N,v,\mP)$ is the standard outside option claim of player $j$, characterized in Section \ref{SEC_OOCLAIM}.

To establish some basic properties of the $\alpha$-value, we first introduce a more general construction.
For convenience, we call a function $\delta$ that maps each pair $(N,v)$ to a payoff vector $\delta(N,v)\in \mR^N$ a \textit{TU-value}.
The Shapley value is an example of a TU-value.
We consider the following standard properties of TU-values.
\begin{itemize}
\item[] \textbf{Efficiency}: For every $N\in \mN$ and $v\in\mG(N)$, $\sum_{j\in N}\delta_j(N,v)=v(N)$.
\item[] \textbf{Symmetry}: For every $N\in \mN$, $v\in\mG(N)$, and $i,j\in N$, if $i \overset{v}{\sim} j$, then $\delta_i(N,v)=\delta_j(N,v)$.
\item[] \textbf{Additivity}: For every $N\in \mN$ and $v,w\in\mG(N)$, $\delta(N,v+w)=\delta(N,v)+\delta(N,w)$.
\item[] \textbf{Null Player}: For every $N\in \mN$ and $v\in\mG(N)$, if $i\in N$ is a null player in $v$, then $\delta_i(N,v)=0$.
\item[] \textbf{Null Player Out}: For every $N\in \mN$ with $|N|\geq 2$ and $v\in\mG(N)$, if $k\in N$ is a null player in $v$, then $\delta_i(N,v)=\delta_i(N\setminus \{k\},v^{-k})$ for every $i\in N\setminus \{k\}$.
\end{itemize}

The following proposition generalizes the construction underlying the $\alpha$-value to an arbitrary TU-value $\delta$. The $\alpha$-value is recovered when $\delta=\mSH$.
It shows how the induced CS-value $f^\delta$ inherits the properties of $\delta$.

\begin{proposition}\label{PROP_TU2CS}
Let $\delta$ be a TU-value. Define the CS-value $f^\delta$ as follows: for every $(N,v,\mP)\in \Gamma$ and $i\in N$,
\begin{align*}
f^\delta_i(N,v,\mP)=\delta_i(\mP(i), v|_{\mP(i)}) 
&+ |\mP(i)| \cdot (\delta_i(N,v) - \delta_i(\mP(i), v|_{\mP(i)})) \\
&- \sum_{j\in \mP(i)}(\delta_j(N,v) - \delta_j(\mP(i), v|_{\mP(i)})).
\end{align*}
Then the following statements hold:
\begin{itemize}
\item[i.] If $\delta$ satisfies Efficiency, then $f^\delta$ satisfies \mCE. 
\item[ii.] If $\delta$ satisfies Symmetry, then $f^\delta$ satisfies \mCS. 
\item[iii.] If $\delta$ satisfies Additivity, then $f^\delta$ satisfies \mADD.
\item[iv.] If $\delta$ satisfies Null Player and Null Player Out, then $f^\delta$ satisfies \mFOT.
\end{itemize}
\end{proposition}

Setting $\delta=\mSH$ in Proposition \ref{PROP_TU2CS} shows that the $\alpha$-value satisfies \mCE, \mCS, \mADD, and \mFOT. By Proposition \ref{PROP_CE_FOT_GN}, the $\alpha$-value also satisfies \mGN. Therefore, the $\alpha$-value coincides with the Shapley value when $\mP=\{N\}$.
Moreover, the following proposition shows that the $\alpha$-value is the unique CS-value satisfying these properties.

\begin{proposition}\label{PROP_ALPHA}
A CS-value $f$ satisfies \mCE, \mCS, \mADD, and \mFOT  if and only if $f=\alpha$.
\end{proposition}

The following result provides the opposite benchmark. The $\alpha$-value and the \mAD-value can be interpreted as two extreme CS-values: the $\alpha$-value fully incorporates outside option claims, while the \mAD-value does not incorporate them.

\begin{proposition}\label{PROP_AD}
A CS-value $f$ satisfies \mCE, \mCS, \mADD, and \mZOT  if and only if $f=\mAD$.
\end{proposition}

Propositions \ref{PROP_ALPHA} and \ref{PROP_AD} therefore characterize the two CS-values through contrasting requirements on outside option claims.
Moreover, the following result locates the $\chi$-value between these two benchmarks: it is a component-wise convex combination of the $\alpha$-value and the \mAD-value.

\begin{proposition}\label{PROP_CONV_COMB}
For every $(N,v,\mP)\in\Gamma$ and $i\in N$, 
\[
\chi_i(N,v,\mP)=\frac{1}{|\mP(i)|}\alpha_i(N,v,\mP)+(1-\frac{1}{|\mP(i)|})\mAD_i(N,v,\mP).
\]
\end{proposition}

The weight assigned to the $\alpha$-value in this convex combination, namely, $1/|\mP(i)|$, depends on the component to which player $i$ belongs.
Therefore, all players in the same component, say $C\in\mP$, receive the same weight $1/|C|$, whereas a player in another component $D\in\mP$ has weight $1/|D|$, which may differ from $1/|C|$.

Proposition \ref{PROP_CONV_COMB} clarifies how the $\chi$-value incorporates outside option claims. Since the \mAD-value does not incorporate outside option claims, whereas the $\alpha$-value incorporates them fully, the $\chi$-value incorporates outside option claims in a discounted form, with the weight depending on the size of the player's component.
In particular, the larger the component to which a player belongs, the smaller the fraction of the full outside option adjustment incorporated into the player's payoff.

\section{Conclusion}\label{SEC_CONC}

We have examined how outside options are reflected in payoff distributions in games with coalition structures. Our main findings are as follows. 

\begin{enumerate}

\item We defined the standard outside option claim rule, which quantifies each player's outside option claim in a CS-game, and provided an axiomatic characterization of the rule.

\item We introduced and characterized the $\alpha$-value. We also provided an alternative characterization of the \mAD-value. These results show that the $\alpha$-value fully incorporates outside option claims, whereas the \mAD-value does not incorporate them.

\item We showed that the $\chi$-value can be expressed as a component-wise convex combination of the $\alpha$-value and the \mAD-value. This result implies that the $\chi$-value incorporates outside option claims in a discounted form, with the weight depending on the size of the player's component.
The following formulas make explicit the difference between the $\chi$-value and the $\alpha$-value:
\begin{itemize}
\item[] $\chi_i(N,v,\mP)=\mAD_i(N,v,\mP)+\frac{1}{|\mP(i)|}\left( |\mP(i)|\cdot \omega^*_i(N,v,\mP)-\sum_{j\in \mP(i)}\omega^*_j(N,v,\mP) \right)$,
\item[] $\alpha_i(N,v,\mP)= \mAD_i(N,v,\mP) + \left( |\mP(i)| \cdot \omega^*_i(N,v,\mP) - \sum_{j\in \mP(i)}\omega^*_j(N,v,\mP) \right)$.
\end{itemize}
Both values use the \mAD-value as the baseline allocation, but the $\chi$-value discounts the outside option adjustment by the size of the player's component.
\end{enumerate}

The characterization results are summarized in Table \ref{TABLE_AXIOM}.
The symbols $+$ and $-$ indicate that the corresponding CS-value satisfies and violates the axiom, respectively. The symbol $\oplus$ indicates that the axiom is used to characterize the corresponding CS-value.
\begin{table}[htbp]
\centering
\caption{Characterization results} \label{TABLE_AXIOM}
\begin{tabular}{ccccccccc}
& \mCE & \mCS & \mADD & \mGN & \mSP & \mOU & \mZOT & \mFOT \\
\hline
$\chi$ & $\oplus$ & $\oplus$ & $\oplus$ & $\oplus$ & $\oplus$ & $-$ & $-$ & $-$ \\
$W$ & $\oplus$ & $\oplus$ & $\oplus$ & $+$ & $-$ & $\oplus$ & $-$ & $-$ \\
$\mAD$ & $\oplus$ & $\oplus$ & $\oplus$ & $+$ & $-$ & $-$ & $\oplus$ & $-$  \\
$\alpha$ & $\oplus$ & $\oplus$ & $\oplus$ & $+$ & $-$ & $-$ & $-$ & $\oplus$ \\
\hline
\end{tabular}
\end{table}

\section{Proofs}\label{SEC_PROOFS}
\subsection*{Proof of Proposition \ref{PROP_OMEGA}}
\begin{proof}
The OO-rule $\omega^*(N,v,\mP)=\mSH(N,v)-\mAD(N,v,\mP)$ satisfies \mZG, \mOE, \mZI, and \mEI.
We show its uniqueness.
Suppose that an OO-rule $\omega$ satisfies \mZG, \mOE, \mZI, and \mEI.

Let $N\in \mN$ and $\mP\in\Pi(N)$.
Let $T\subseteq N$ ($T\neq \emptyset$), $\lambda\in\mR$, and $w\in\mG(N)$.
By \mOE, we have
$\sum_{j\in N}\omega_j(N,w+\lambda u_T,\mP)
=(w+\lambda u_T)(N) -\sum_{C\in\mP}(w+\lambda u_T)(C)
=w(N)-\sum_{C\in\mP}w(C) + \left(\lambda u_T(N)-\sum_{C\in\mP}\lambda u_T(C)\right)$
and $\sum_{j\in N}\omega_j(N,w,\mP)=w(N)-\sum_{C\in\mP}w(C)$.
Hence, we have
\begin{align}
\sum_{j\in N} \left(\omega_j(N,w+\lambda u_T,\mP) - \omega_j(N,w,\mP) \right)
&=\lambda u_T(N)-\sum_{C\in\mP}\lambda u_T(C) \nonumber\\
&=\lambda \left(1 - \sum_{C\in\mP}u_T(C) \right). \label{eq_0713_1137}
\end{align}
\mZI implies that for every $i\in N\setminus T$,
\begin{equation}
\omega_i(N,w+\lambda u_T,\mP) - \omega_i(N,w,\mP)=0. \label{eq_0713_1140}
\end{equation}
\mEI implies that $\omega_i(N,w+\lambda u_T,\mP) - \omega_i(N,w,\mP)=\omega_j(N,w+\lambda u_T,\mP) - \omega_j(N,w,\mP)$ for every $i,j\in T$.
Hence, by (\ref{eq_0713_1137}) and (\ref{eq_0713_1140}), for every $i\in T$,
\begin{equation}
\omega_i(N,w+\lambda u_T,\mP) - \omega_i(N,w,\mP) = \frac{\lambda}{|T|}\left(1 - \sum_{C\in\mP}u_T(C) \right). \label{eq_0713_1141}
\end{equation}

First, suppose that $|\mP(T)|=1$.
There is unique $C\in\mP$ such that $T\subseteq C$.
Therefore, $\sum_{C\in\mP}u_T(C)=1$.
It follows from (\ref{eq_0713_1140}) and (\ref{eq_0713_1141}) that for every $i\in N$, $\omega_i(N,w+\lambda u_T,\mP) - \omega_i(N,w,\mP)=0$.
Next, suppose that $|\mP(T)|\geq 2$.
Then, $T$ is not contained in any component of $\mP$. Hence, $\sum_{C\in\mP}u_T(C)=0$.
By (\ref{eq_0713_1141}), for every $i\in T$,
$\omega_i(N,w+\lambda u_T,\mP) - \omega_i(N,w,\mP)=\frac{\lambda}{|T|}$.
By (\ref{eq_0713_1140}), for every $i\in N\setminus T$, $\omega_i(N,w+\lambda u_T,\mP) - \omega_i(N,w,\mP)=0$.
Hence, it holds that for every $i\in N$,
\begin{align}
\omega_i(N,w+\lambda u_T,\mP)-\omega_i(N,w,\mP)=
\begin{cases}
\frac{\lambda}{|T|} & \text{if $i\in T$ and $|\mP(T)|\geq 2$},\\
0 & \text{otherwise}.
\end{cases}
\label{eq_0713_1151}
\end{align}

For every $v\in \mG(N)$, we have $v=\sum_{\emptyset\neq T\subseteq N}\lambda_T^v u_T$.
Starting from the zero game $\textbf{0}$, we repeatedly add $\lambda_T^v u_T$.
\mZG and (\ref{eq_0713_1151}) imply that, for every $i\in N$,
\begin{equation}
\omega_i(N,v,\mP) =\sum_{\substack{T\subseteq N\\ i\in T,\ |\mP(T)|\geq 2}} \frac{\lambda_T^v}{|T|}. \label{eq_0713_1434}
\end{equation}
We have
$\mSH_i(N,v)=\sum_{\substack{T\subseteq N\\ i\in T}} \frac{\lambda_T^v}{|T|}$
and
$\mAD_i(N,v,\mP)=\mSH_i(\mP(i),v|_{\mP(i)})
=\sum_{\substack{T\subseteq N\\ i\in T,\ |\mP(T)|=1}} \frac{\lambda_T^v}{|T|}$.
Hence, (\ref{eq_0713_1434}) implies that $\omega_i(N,v,\mP)=\mSH_i(N,v)-\mAD_i(N,v,\mP)=\omega_i^*(N,v,\mP)$.
\end{proof}

\subsection*{Proof of Proposition \ref{PROP_CE_FOT_GN}}
\begin{proof}
Let $f$ satisfy \mCE.
Let $N\in \mN$ and $v\in \mG(N)$.
Let $k\in N$ be a null player in $v$.
We show that $f_k(N,v,\{N\})=0$.
If $|N|=1$, then $v(N)=v(\{k\})=0$. Hence, by \mCE, $f_k(N,v,\{N\})=0$.
Suppose that $|N|\geq 2$.
If $f$ satisfies \mFOT, then (setting $\mP:=\{N\}$) for every $i\in N\setminus \{k\}$, we have
$f_i(N,v,\{N\})-f_i(N\setminus \{k\},v^{-k},\{N\setminus \{k\}\})=f_i(N,v,\{N\})-f_i(N,v|_{N},\{N\})=0$.
Hence, for every $i\in N\setminus \{k\}$,
\begin{equation}
f_i(N,v,\{N\})=f_i(N\setminus \{k\},v^{-k},\{N\setminus \{k\}\}). \label{eq_0628_1603}
\end{equation}
If $f$ satisfies \mZOT, then (\ref{eq_0628_1603}) holds directly.
Hence, (\ref{eq_0628_1603}) holds in either case.
Now, we have 
\[
v(N)\overset{\mCE}{=}\sum_{j\in N}f_j(N,v,\{N\})
\]
and
\begin{eqnarray*}
v(N\setminus \{k\})
&=&v^{-k}(N\setminus \{k\}) \\
&\overset{\mCE}{=}& \sum_{j\in N\setminus \{k\}}f_j(N\setminus \{k\},v^{-k},\{N\setminus \{k\}\}) \\
&\overset{(\ref{eq_0628_1603})}{=}& \sum_{j\in N\setminus \{k\}}f_j(N,v,\{N\}).
\end{eqnarray*}
Since $k$ is a null player in $v$, we have $v(N)=v(N\setminus \{k\})$, which implies that $\sum_{j\in N}f_j(N,v,\{N\}) = \sum_{j\in N\setminus \{k\}}f_j(N,v,\{N\})$.
Therefore, $f_k(N,v,\{N\})=0$.
\end{proof}

\subsection*{Proof of Proposition \ref{PROP_TU2CS}}
\begin{proof}
For notational simplicity, we write $\delta^{S}_i := \delta_i(S, v|_{S})$ for every $S\subseteq N$ with $i\in S$.

(i) Let $\delta$ satisfy Efficiency. Let $(N,v,\mP)\in \Gamma$. For every $C\in \mP$, we have
\begin{align*}
\sum_{j\in C}f^\delta_j(N,v,\mP)=
\sum_{j\in C} \delta_j^{C} 
+ |C| \cdot \sum_{j\in C}(\delta^{N}_j - \delta^{C}_j) - |C| \cdot \sum_{j\in C}(\delta^{N}_j - \delta^{C}_j)
=\sum_{j\in C} \delta_j^{C}.
\end{align*}
Since $\delta$ satisfies Efficiency, we obtain $\sum_{j\in C} \delta_j^{C}=v(C)$.

(ii) Let $\delta$ satisfy Symmetry. Let $(N,v,\mP)\in \Gamma$, $C\in \mP$, and $i,j\in C$ with $i \overset{v}{\sim} j$. We have
\begin{align*}
&f^\delta_i(N,v,\mP)=\delta^{C}_i + |C| \cdot (\delta^{N}_i - \delta^{C}_i) - \sum_{\ell\in C}(\delta^{N}_\ell - \delta^{C}_\ell), \text{ and}\\
&f^\delta_j(N,v,\mP)=\delta^{C}_j + |C| \cdot (\delta^{N}_j - \delta^{C}_j) - \sum_{\ell\in C}(\delta^{N}_\ell - \delta^{C}_\ell).
\end{align*}
Since $i \overset{v}{\sim} j$ and $i,j\in C$, it holds that $i \overset{v|_{C}}{\sim} j$.
Since $\delta$ satisfies Symmetry, $\delta^{N}_i=\delta^{N}_j$ follows from $i \overset{v}{\sim} j$, and $\delta^{C}_i=\delta^{C}_j$ follows from $i \overset{v|_{C}}{\sim} j$.
Thus, $f^\delta_i(N,v,\mP)=f^\delta_j(N,v,\mP)$.

(iii) Let $\delta$ satisfy Additivity. Let $N\in \mN$, $\mP\in\Pi(N)$, and $v,w\in\mG(N)$.
For every $C\in \mP$ and $S\subseteq C$, we have $(v+w)|_{C}(S)\overset{S\subseteq C}{=}(v+w)(S)=v(S)+w(S)\overset{S\subseteq C}{=}v|_{C}(S)+w|_{C}(S)$. Hence, for every $C\in \mP$, $(v+w)|_{C} = v|_{C} + w|_{C}$.
Since $\delta$ satisfies Additivity, it holds that for every $i\in N$, $\delta_i(\mP(i), (v+w)|_{\mP(i)})=\delta_i(\mP(i), v|_{\mP(i)})+\delta_i(\mP(i), w|_{\mP(i)})$.
Hence, for every $i\in N$, 
\begin{align*}
f^\delta_i(N,v+w,\mP)
=&\delta_i(\mP(i), (v+w)|_{\mP(i)}) \\
&+ |\mP(i)| \cdot (\delta_i(N,v+w) - \delta_i(\mP(i), (v+w)|_{\mP(i)})) \\
&- \sum_{j\in \mP(i)}(\delta_j(N,v+w) - \delta_j(\mP(i), (v+w)|_{\mP(i)})) \\
=&\delta_i(\mP(i), v|_{\mP(i)}) + \delta_i(\mP(i), w|_{\mP(i)}) \\
&+ |\mP(i)| \cdot (\delta_i(N,v) + \delta_i(N,w) - \delta_i(\mP(i), v|_{\mP(i)}) - \delta_i(\mP(i), w|_{\mP(i)})) \\
&- \sum_{j\in \mP(i)}(\delta_j(N,v) + \delta_j(N,w) - \delta_j(\mP(i), v|_{\mP(i)}) - \delta_j(\mP(i), w|_{\mP(i)})) \\
=&f_i^\delta(N,v,\mP)+f_i^\delta(N,w,\mP).
\end{align*}
Thus, $f^\delta$ satisfies \mADD.

(iv) Let $\delta$ satisfy Null Player and Null Player Out.
Let $(N,v,\mP)\in \Gamma$ with $|N|\geq 2$. Suppose that $k\in N$ is a null player in $v$ and $|\mP(k)|\geq 2$.
We show that for every $i\in \mP(k)\setminus \{k\}$, $f^\delta_i(N,v,\mP)-f^\delta_i(N\setminus \{k\},v^{-k},\mP^{-k})=f^\delta_i(N,v,\{N\})-f^\delta_i(\mP(i),v|_{\mP(i)},\{\mP(i)\})$.

Let $i\in \mP(k)\setminus \{k\}$. We have
\begin{align}
&f^\delta_i(N,v,\mP) - f^\delta_i(N\setminus \{k\},v^{-k},\mP^{-k}) \nonumber\\
&=\delta^{\mP(i)}_i + |\mP(i)| \cdot (\delta^{N}_i - \delta^{\mP(i)}_i) - \sum_{j\in \mP(i)}(\delta^{N}_j - \delta^{\mP(i)}_j) \nonumber\\
&-\delta^{\mP(i)\setminus \{k\}}_i - (|\mP(i)|-1) \cdot (\delta^{N\setminus \{k\}}_i - \delta^{\mP(i)\setminus \{k\}}_i) + \sum_{j\in \mP(i)\setminus \{k\}}(\delta^{N\setminus \{k\}}_j - \delta^{\mP(i)\setminus \{k\}}_j). \label{eq_0629_0954}
\end{align}
Since $k\in \mP(i)$ is a null player in $v$, $k$ is also a null player in $v|_{\mP(i)}$.
Since $\delta$ satisfies Null Player Out, it holds that $\delta^{N}_j=\delta^{N\setminus \{k\}}_j$ and $\delta^{\mP(i)}_j=\delta^{\mP(i)\setminus \{k\}}_j$ for every $j\in \mP(i)\setminus \{k\}$.
Hence, we have
\begin{align}
(\ref{eq_0629_0954})=(\delta^{N\setminus \{k\}}_i - \delta^{\mP(i)\setminus \{k\}}_i)-(\delta^{N}_k - \delta^{\mP(i)}_k). \label{eq_0629_0958}
\end{align}
Since $\delta$ satisfies Null Player, we have $\delta^{N}_k = \delta^{\mP(i)}_k =0$. Therefore, we obtain
\begin{equation}
(\ref{eq_0629_0958})
=\delta^{N\setminus \{k\}}_i - \delta^{\mP(i)\setminus \{k\}}_i=\delta^{N}_i - \delta^{\mP(i)}_i, \label{eq_0629_1007}
\end{equation}
where the last equality follows from Null Player Out.

We now consider the right-hand side of the equation of \mFOT. It holds that
\begin{align*}
&f^\delta_i(N,v,\{N\}) - f^\delta_i(\mP(i),v|_{\mP(i)},\{\mP(i)\}) \\
&=\delta^{N}_i + |N| \cdot (\delta^{N}_i - \delta^{N}_i) - \sum_{j\in N}(\delta^{N}_j - \delta^{N}_j) \\
&-\delta^{\mP(i)}_i - |\mP(i)| \cdot (\delta^{\mP(i)}_i - \delta^{\mP(i)}_i) + \sum_{j\in \mP(i)}(\delta^{\mP(i)}_j - \delta^{\mP(i)}_j) \\
&=\delta^{N}_i - \delta^{\mP(i)}_i.
\end{align*}
This is equal to (\ref{eq_0629_1007}).
\end{proof}

\subsection*{Proof of Proposition \ref{PROP_ALPHA}}
\begin{proof}
Sufficiency: Setting $\delta=\mSH$, the function $f^\delta$ of Proposition \ref{PROP_TU2CS} coincides with the $\alpha$-value.
Since the Shapley value satisfies Efficiency, Symmetry, Additivity, Null Player, and Null Player Out, Proposition \ref{PROP_TU2CS} implies that the $\alpha$-value satisfies \mCE, \mCS, \mADD, and \mFOT.

Uniqueness: Let $f$ satisfy \mCE, \mCS, \mADD, and \mFOT.
Since $f$ satisfies \mCE and \mFOT, Proposition \ref{PROP_CE_FOT_GN} implies that $f$ satisfies \mGN.
Let $(N,v,\mP)\in \Gamma$.
Since $v$ is uniquely expressed as $v=\sum_{T\in 2^N\setminus \{\emptyset\}} \lambda^v_T u_T$, \mADD implies that $f(N,v,\mP)=\sum_{T\in 2^N\setminus \{\emptyset\}} f(N, \lambda^v_T u_T, \mP)$. Hence, we show that for every $\lambda\in\mR$ and nonempty $T\subseteq N$, $f(N,\lambda u_T,\mP)$ is uniquely determined.
Let $\lambda\in\mR$ and $T\subseteq N$ with $T\neq \emptyset$. Consider the following two cases: $|\mP(T)|=1$ and $|\mP(T)|\geq 2$.
\\

Case $|\mP(T)|=1$:
Let $C\in \mP(T)$. We have $T\subseteq C$.
If $T=C$, then \mCE and \mCS imply that for every $i\in C(=T)$, $f_i(N, \lambda u_T, \mP)=\frac{\lambda}{|T|}$.
Therefore, consider $T\subsetneq C$. Let $K:=\{k_1,...,k_m\}=C\setminus T$.
Let $i\in T$. Hence, $|C|\geq |K|+1$.
Since $k_1$ is a null player in $\lambda u_T$, we have
\begin{equation}
f_i(N,\lambda u_T,\mP)-f_i(N\setminus \{k_1\},\lambda u_T^{-k_1},\mP^{-k_1})
\overset{\mFOT}{=}f_i(N,\lambda u_T,\{N\})-f_i(C,\lambda u_T|_{C},\{C\}). \label{eq_0702_1001}
\end{equation}
From \mGN, \mCE, \mCS, and $i\in T$, it follows that $f_i(N,\lambda u_T,\{N\})= \frac{\lambda}{|T|}$ and $f_i(C,\lambda u_T|_{C},\{C\})=\frac{\lambda}{|T|}$.
Hence, by (\ref{eq_0702_1001}), we have
\[
f_i(N,\lambda u_T,\mP)-f_i(N\setminus \{k_1\},\lambda u_T^{-k_1},\mP^{-k_1})=0.
\]
Similarly, we have
\[
f_i(N\setminus \{k_1\},\lambda u_T^{-k_1},\mP^{-k_1})- f_i(N\setminus \{k_1,k_2\}, \lambda u_T^{-\{k_1,k_2\}}, \mP^{-\{k_1,k_2\}}) =0.
\]
Repeating this, we have
\[
f_i(N\setminus \{k_1,...,k_{m-1}\}, \lambda u_T^{-\{k_1, ..., k_{m-1}\}}, \mP^{-\{k_1, ..., k_{m-1}\}}) - f_i(N\setminus K,\lambda u_T^{-K},\mP^{-K}) =0.
\]
Summing up the equations, we obtain
\[
f_i(N,\lambda u_T,\mP) - f_i(N\setminus K,\lambda u_T^{-K},\mP^{-K}) =0,
\]
where, by \mCE and \mCS, $f_i(N\setminus K,\lambda u_T^{-K},\mP^{-K})=\frac{\lambda}{|T|}$.
Hence, $f_i(N,\lambda u_T,\mP)=\frac{\lambda}{|T|}$. This holds for every $i\in T$.
Therefore, \mCE and \mCS imply that $f_i(N,\lambda u_T,\mP)=0$ for every $i\in C\setminus T$.
Note that by \mCE and \mCS, $f_i(N, \lambda u_T, \mP)=0$ for every $i\in N\setminus C$.
\\

Case $|\mP(T)|\geq 2$:
Let $C\in \mP(T)$. Since $|\mP(T)|\geq 2$, $T\setminus C\neq \emptyset$ (equivalently, $T\not\subseteq C$).
If $C\subseteq T$, then \mCE and \mCS imply that $f_i(N,\lambda u_T,\mP)=0$ for every $i\in C$.
Hence, consider $C\in \mP(T)$ with $C\not\subseteq T$. Hence, $C\setminus T\neq \emptyset$.
Let $K:=\{k_1,...,k_m\}=C\setminus T$.
Let $i\in C\cap T$. Hence, $|C|\geq |K|+1$.
Since $k_1$ is a null player in $\lambda u_T$, we have
\begin{equation}
f_i(N,\lambda u_T,\mP)-f_i(N\setminus \{k_1\},\lambda u_T^{-k_1},\mP^{-k_1})
\overset{\mFOT}{=}f_i(N,\lambda u_T,\{N\})-f_i(C,\lambda u_T|_{C},\{C\}). \label{eq_0702_1003}
\end{equation}
From \mGN, \mCE, \mCS, and $i\in C\cap T$, it follows that $f_i(N,\lambda u_T,\{N\})= \frac{\lambda}{|T|}$.
Since $T\not\subseteq C$, we have $u_T|_{C}(S)=0$ for every $S\subseteq C$.
Hence, all members of $C$ are null players in $u_T|_{C}$. 
By \mGN, $f_i(C,\lambda u_T|_{C},\{C\})=0$.
Therefore, by (\ref{eq_0702_1003}), we have
\[
f_i(N,\lambda u_T,\mP)-f_i(N\setminus \{k_1\},\lambda u_T^{-k_1},\mP^{-k_1})=\frac{\lambda}{|T|}.
\]
Similarly, we have
\[
f_i(N\setminus \{k_1\},\lambda u_T^{-k_1},\mP^{-k_1})- f_i(N\setminus \{k_1,k_2\}, \lambda u_T^{-\{k_1,k_2\}}, \mP^{-\{k_1,k_2\}}) =\frac{\lambda}{|T|}.
\]
Repeating this, we have
\[
f_i(N\setminus \{k_1,...,k_{m-1}\}, \lambda u_T^{-\{k_1, ..., k_{m-1}\}}, \mP^{-\{k_1, ..., k_{m-1}\}}) - f_i(N\setminus K,\lambda u_T^{-K},\mP^{-K}) =\frac{\lambda}{|T|}.
\]
Summing up the equations, we obtain
\[
f_i(N,\lambda u_T,\mP) - f_i(N\setminus K,\lambda u_T^{-K},\mP^{-K}) =|K|\cdot\frac{\lambda}{|T|}=\lambda \cdot\frac{|C\setminus T|}{|T|},
\]
where, by \mCE and \mCS, $f_i(N\setminus K,\lambda u_T^{-K},\mP^{-K})=0$.
Hence, $f_i(N,\lambda u_T,\mP)=\lambda \frac{|C\setminus T|}{|T|}$. 
Therefore, for every $C\in \mP(T)$, \mCE and \mCS imply that $f_i(N,\lambda u_T,\mP)=\lambda \frac{|C\setminus T|}{|T|}$ for every $i\in C\cap T$, and $f_i(N,\lambda u_T,\mP)=-\lambda \frac{|C\cap T|}{|T|}$ for every $i\in C\setminus T$.
Note that by \mCE and \mCS, $f_i(N,\lambda u_T,\mP)=0$ for every $i\in N\setminus (\cup_{C'\in \mP(T)}C')$.
\\

From both cases, it follows that
for every $i\in \cup_{C'\in \mP(T)}C'$,
\[
f_i(N,\lambda u_T,\mP)=
\begin{cases}
\frac{\lambda}{|T|} & \text{if $|\mP(T)|=1$ and $i\in T$},\\
0 & \text{if $|\mP(T)|=1$ and $i\notin T$},\\
\lambda \frac{|\mP(i)\setminus T|}{|T|} & \text{if $|\mP(T)|\geq 2$ and $i\in T$},\\
-\lambda \frac{|\mP(i)\cap T|}{|T|} & \text{if $|\mP(T)|\geq 2$ and $i\notin T$},
\end{cases}
\]
and $f_i(N,\lambda u_T,\mP)=0$ for every $i\in N\setminus \cup_{C'\in \mP(T)}C'$.
This coincides with $\alpha(N,\lambda u_T,\mP)$.
Hence, $f(N,v,\mP)=\alpha(N,v,\mP)$.
\end{proof}

\subsection*{Proof of Proposition \ref{PROP_AD}}
\begin{proof}
Sufficiency: The \mAD-value satisfies \mCE, \mCS, and \mADD.
Let $(N,v,\mP)\in \Gamma$ with $|N|\geq 2$. Let $k\in N$ be a null player in $v$ and $|\mP(k)|\geq 2$. Since the Shapley value satisfies Null Player Out, for every $i\in \mP(k)\setminus \{k\}$, we have $\mAD_i(N,v,\mP)=\mSH_i(\mP(i), v|_{\mP(i)})
=\mSH_i(\mP(i)\setminus \{k\}, v|_{\mP(i)\setminus \{k\}})
=\mAD_i(N\setminus \{k\},v^{-k},\mP^{-k})$.
Hence, the \mAD-value satisfies \mZOT.

Uniqueness: Let $f$ satisfy \mCE, \mCS, \mADD, and \mZOT.
Let $(N,v,\mP)\in \Gamma$.
Since $v$ is uniquely expressed as $v=\sum_{T\in 2^N\setminus \{\emptyset\}} \lambda^v_T u_T$, \mADD implies that $f(N,v,\mP)=\sum_{T\in 2^N\setminus \{\emptyset\}} f(N, \lambda^v_T u_T, \mP)$. Hence, we show that for every $\lambda\in\mR$ and nonempty $T\subseteq N$, $f(N,\lambda u_T,\mP)$ is uniquely determined.
Let $\lambda\in\mR$ and $T\subseteq N$ with $T\neq \emptyset$. Consider the following two cases: $|\mP(T)|=1$ and $|\mP(T)|\geq 2$.
\\

Case $|\mP(T)|=1$:
Let $C\in \mP(T)$. We have $T\subseteq C$.
If $T=C$, then \mCE and \mCS imply that for every $i\in C(=T)$, $f_i(N, \lambda u_T, \mP)=\frac{\lambda}{|T|}$.
Therefore, consider $T\subsetneq C$. Let $K:=\{k_1,...,k_m\}=C\setminus T$.
Let $i\in T$. Hence, $|C|\geq |K|+1$.
Since $k_1,...,k_m$ are null players in $\lambda u_T$, \mZOT implies that
\begin{eqnarray}
f_i(N,\lambda u_T,\mP)
&=&f_i(N\setminus \{k_1\},\lambda u_T^{-k_1},\mP^{-k_1})\nonumber\\
&=&f_i(N\setminus \{k_1,k_2\}, \lambda u_T^{-\{k_1,k_2\}}, \mP^{-\{k_1,k_2\}})\nonumber\\
&...& \nonumber\\
&=&f_i(N\setminus K,\lambda u_T^{-K},\mP^{-K})\nonumber\\
&\overset{\mCE,\mCS}{=}&\frac{\lambda}{|T|}\label{eq_0702_1017}
\end{eqnarray}
Therefore, $f_i(N,\lambda u_T,\mP)=\frac{\lambda}{|T|}$ for every $i\in T$ and, by \mCE and \mCS, $f_i(N,\lambda u_T,\mP)=0$ for every $i\in C\setminus T$.
Note that by \mCE and \mCS, $f_i(N, \lambda u_T, \mP)=0$ for every $i\in N\setminus C$.
\\

Case $|\mP(T)|\geq 2$:
Let $C\in \mP(T)$. Since $|\mP(T)|\geq 2$, $T\setminus C\neq \emptyset$ (equivalently, $T\not\subseteq C$).
If $C\subseteq T$, then \mCE and \mCS imply that $f_i(N,\lambda u_T,\mP)=0$ for every $i\in C$.
Hence, consider $C\in \mP(T)$ with $C\not\subseteq T$. Hence, $C\setminus T\neq \emptyset$.
Let $K:=\{k_1,...,k_m\}=C\setminus T$. Since $C\in \mP(T)$, $|C\cap T|\geq 1$. Hence, $|C|\geq |K|+1$.
Since $k_1,...,k_m$ are null players in $\lambda u_T$, in the same manner as (\ref{eq_0702_1017}), for every $i\in C\cap T$,
\[
f_i(N,\lambda u_T,\mP)\overset{\mZOT}{=}f_i(N\setminus K,\lambda u_T^{-K},\mP^{-K})\overset{\mCE,\mCS}{=}0.
\]
Therefore, by \mCE and \mCS, $f_i(N,\lambda u_T,\mP)=0$ for every $i\in C\setminus T$.
Moreover, \mCE and \mCS imply that $f_i(N,\lambda u_T,\mP)=0$ for every $i\in N\setminus (\cup_{C'\in \mP(T)}C')$.
\\

From both cases, it follows that
for every $i\in \cup_{C'\in \mP(T)}C'$,
\[
f_i(N,\lambda u_T,\mP)=
\begin{cases}
\frac{\lambda}{|T|} & \text{if $|\mP(T)|=1$ and $i\in T$},\\
0 & \text{if $|\mP(T)|=1$ and $i\notin T$},\\
0 & \text{if $|\mP(T)|\geq 2$ and $i\in T$},\\
0 & \text{if $|\mP(T)|\geq 2$ and $i\notin T$},
\end{cases}
\]
and $f_i(N,\lambda u_T,\mP)=0$ for every $i\in N\setminus (\cup_{C'\in \mP(T)}C')$.
This coincides with $\mAD(N,\lambda u_T,\mP)$.
Hence, $f(N,v,\mP)=\mAD(N,v,\mP)$.
\end{proof}

\subsection*{Proof of Proposition \ref{PROP_CONV_COMB}}
\begin{proof}
Let $(N,v,\mP)\in\Gamma$ and $i\in N$. 
Since $v(\mP(i))=\sum_{j\in \mP(i)}\mAD_j(N,v,\mP)$, we have
\[
\chi_i(N,v,\mP)
=\mSH_i(N,v)+\frac{1}{|\mP(i)|}\sum_{j\in \mP(i)}(\mAD_j(N,v,\mP)-\mSH_j(N,v)).
\]
Since $\omega^*_j(N,v,\mP)=\mSH_j(N,v)-\mAD_j(N,v,\mP)$ for every $j\in N$, we have
\[
\chi_i(N,v,\mP)
=\mAD_i(N,v,\mP)+\omega^*_i(N,v,\mP)-\frac{1}{|\mP(i)|}\sum_{j\in \mP(i)}\omega^*_j(N,v,\mP).
\]
From
\[
\alpha_i(N,v,\mP)
=\mAD_i(N,v,\mP)
+|\mP(i)|\cdot\omega^*_i(N,v,\mP)-\sum_{j\in \mP(i)}\omega^*_j(N,v,\mP),
\]
it follows that
\[
\chi_i(N,v,\mP)=\frac{1}{|\mP(i)|}\alpha_i(N,v,\mP)+(1-\frac{1}{|\mP(i)|})\mAD_i(N,v,\mP).
\]
\end{proof}

\section{Independence of Axioms}\label{SEC_INDEP}

\subsection{Independence of Axioms of Proposition \ref{PROP_OMEGA}}
\begin{itemize}
\item Violating \mZG. Let $\sigma$ be an ordering of all players in $\mathcal{U}$. For every finite $T\subseteq \mathcal{U}$ with $|T|\geq 2$, let $1^{\sigma}(T)$ ($2^{\sigma}(T)$) denote the first (second) player in coalition $T$. For every $i\in T$, let
\[
\rho^\sigma_i(T)=
\begin{cases}
1 & \text{if $i=1^{\sigma}(T)$},\\
-1 & \text{if $i=2^{\sigma}(T)$},\\
0 & \text{otherwise}.
\end{cases}
\]
For every $T\subseteq \mathcal{U}$ with $|T|= 1$, let $\rho^\sigma_i(T)=0$.
Define $\kappa^\sigma_i(N,v,\mP):=\omega^*_i(N,v,\mP)+\rho^\sigma_i(N)$.
This function violates \mZG and satisfies the other three axioms. 

\item Violating \mOE. Let $\omega^{0}_i(N,v,\mP):=0$ for every $(N,v,\mP)\in \Gamma$ and $i\in N$. This function violates \mOE and satisfies the other three axioms. 

\item Violating \mZI. Let $\epsilon_i(N,v,\mP):=\mSH_i(N,v)-\frac{1}{|N|}\sum_{C\in \mP}v(C)$ for every $(N,v,\mP)\in \Gamma$ and $i\in N$. This function violates \mZI and satisfies the other three axioms. 

\item Violating \mEI. Let $\sigma$ be an ordering of all players in $\mathcal{U}$. For every $(N,v,\mP)\in \Gamma$ and $i\in N$, let
\[
g^{\sigma}_i(N,v,\mP):=\sum_{\substack{T\subseteq N\\ |\mP(T)|\geq 2,\ 1^{\sigma}(T)=i}}\lambda^v_T.
\]
This function violates \mEI and satisfies the other three axioms. 
\end{itemize}

\subsection{Independence of Axioms of Proposition \ref{PROP_ALPHA}}
\begin{itemize}
\item Violating \mCE. Let $f_i^{0}(N,v,\mP):=0$ for every $(N,v,\mP)\in \Gamma$ and $i\in N$. This function violates \mCE and satisfies the other three axioms. 

\item Violating \mCS. Let $\sigma$ be an ordering of all players in $\mathcal{U}$.
Let $P^{\sigma}_i$ represent the set of predecessors of $i$ in $\sigma$.
For every $N\in \mN$ and $i\in N$, let $P^{\sigma}_i|_N:=P^{\sigma}_i \cap N$.
For every $v\in \mG(N)$ and $i\in N$, let $m^{\sigma}_i(N,v):=v(P^{\sigma}_i|_N \cup \{i\})-v(P^{\sigma}_i|_N)$.
Note that $m^{\sigma}$ is a TU-value that satisfies Efficiency, Null Player, Null Player Out, and Additivity.
Setting $\delta:=m^\sigma$ in Proposition \ref{PROP_TU2CS}, we obtain $f^{m^\sigma}$ that violates \mCS and satisfies the other three axioms. 

\item Violating \mFOT. The \mAD-value violates \mFOT and satisfies the other three axioms. 

\item Violating \mADD. For every $N\in\mN$ and $v\in \mG(N)$, let $h(N,v):=\sum_{j\in N}(\mSH_j(N,v))^2$. For every $i\in N$, define
\[
q_i(N,v)=
\begin{cases}
(\mSH_i(N,v))^2 + \mSH_i(N,v)(1-\frac{h(N,v)}{v(N)}) & \text{if $v(N)\neq 0$},\\
\ \mSH_i(N,v) & \text{if $v(N)= 0$}.\\
\end{cases}
\]
Note that $q$ is a TU-value that satisfies Efficiency, Symmetry, Null Player, and Null Player Out.
Setting $\delta:=q$ in Proposition \ref{PROP_TU2CS}, we obtain $f^{q}$ that violates \mADD and satisfies the other three axioms. 
\end{itemize}

\subsection{Independence of Axioms of Proposition \ref{PROP_AD}}
\begin{itemize}
\item Violating \mCE. The function $f^{0}$ violates \mCE and satisfies the other three axioms. 

\item Violating \mCS. Let $\sigma$ be an ordering of all players in $\mathcal{U}$.
For every $(N,v,\mP)\in \Gamma$ and $i\in N$, let $\mu^{\sigma}_i(N,v,\mP):=v(P^{\sigma}_i|_{\mP(i)} \cup \{i\})-v(P^{\sigma}_i|_{\mP(i)})$.
Function $\mu^\sigma$ violates \mCS and satisfies the other three axioms. 

\item Violating \mZOT. The $\alpha$-value violates \mZOT and satisfies the other three axioms. 

\item Violating \mADD. 
For every $(N,v,\mP)\in \Gamma$ and $i\in N$, let $\eta_i(N,v,\mP)=q_i(\mP(i), v|_{\mP(i)})$.
Function $\eta$ violates \mADD and satisfies the other three axioms. 
\end{itemize}

\section*{Declaration}

\noindent\textbf{Funding.} The author gratefully acknowledge the financial support from JSPS: No.25K16603.

\noindent\textbf{Conflict of interest.} There are no conflicts of interest.

\noindent\textbf{Data availability.} This work is not based on any empirical data. 

\noindent\textbf{Acknowledgments} The author appreciates the helpful comments provided by Rene van den Brink, Youngsub Chun, Yukihiko Funaki, Yukio Koriyama, and Satoshi Nakada.

\end{document}